\documentclass{ecai} 

\usepackage{latexsym}
\usepackage{amssymb}
\usepackage{amsmath}
\usepackage{amsthm}
\usepackage{booktabs}
\usepackage{enumitem}
\usepackage{graphicx}
\usepackage{color}

\usepackage{amsfonts}
\usepackage{algorithm}
\usepackage{algpseudocode}
\usepackage{enumitem}
\usepackage{float}
\usepackage{flushend} 
\usepackage{mathtools}
\usepackage{subfigure}
\usepackage{xcolor}

\newtheorem{theorem}{Theorem}
\newtheorem{lemma}[theorem]{Lemma}

\newtheorem{definition}{Definition}

\newtheorem{example}{Example}

\newcommand{\Auction}{\mathrm{Auct}}

\renewcommand{\Pr}{\mathrm{Pr}}
\newcommand{\R}{\mathbb{R}}
\newcommand{\Stat}{\mathrm{Stat}}
\newcommand{\SecPrice}{\mathrm{SecPrice}}

\algrenewcommand\algorithmicrequire{\textbf{Input:}}
\algrenewcommand\algorithmicensure{\textbf{Output:}}

\newcommand{\BibTeX}{B\kern-.05em{\sc i\kern-.025em b}\kern-.08em\TeX}

\begin{document}


\begin{frontmatter}


\paperid{123} 


\title{Balancing Efficiency with Equality: Auction Design with Group Fairness Concerns}




\author[A]{\fnms{Fengjuan}~\snm{Jia}}
\author[A]{\fnms{Mengxiao}~\snm{Zhang}\thanks{Corresponding author. Email: mengxiao.zhang@uestc.edu.cn}}
\author[B]{\fnms{Jiamou}~\snm{Liu}} 
\author[A]{\fnms{Bakh}~\snm{Khoussainov}} 

\address[A]{School of Computer Science and Engineering, University of Electronic Science and Technology of China, China}
\address[B]{School of Computer Science, The University of Auckland, New Zealand}


\begin{abstract}
The issue of fairness in AI arises from discriminatory practices in applications like job recommendations and risk assessments, emphasising the need for algorithms that do not discriminate based on group characteristics. This concern is also pertinent to auctions, commonly used for resource allocation, which necessitate fairness considerations. Our study examines auctions with groups distinguished by specific attributes, seeking to (1) define a fairness notion that ensures equitable treatment for all, (2) identify mechanisms that adhere to this fairness while preserving incentive compatibility, and (3) explore the balance between fairness and seller's revenue. 
We introduce two fairness notions—group fairness and individual fairness—and propose two corresponding auction mechanisms: the Group Probability Mechanism, which meets group fairness and incentive criteria, and the Group Score Mechanism, which also encompasses individual fairness. Through experiments, we validate these mechanisms' effectiveness in promoting fairness and examine their implications for seller revenue. 
\end{abstract}

\end{frontmatter}

\section{Introduction}
As automated decision-making systems increasingly orchestrate our lives, the issue of algorithmic fairness has become crucial. The 2026 White House report \cite{office2016big} advocates for ``equal opportunity by design'' as a guiding principle in AI system design and development. This principle seeks to address potential disparities in the outcomes of AI systems affecting various socioeconomic groups, defined by race, gender, age, and income brackets. Algorithmic fairness aims to mitigate the impact of these systems, ensuring they neither reinforce existing inequalities nor introduce new forms of discrimination \cite{finocchiaro2021bridging}.


Recent studies reveal pervasive unfairness across multiple domains, highlighting systemic biases in algorithmically driven decision-making systems. Job search systems, for instance, have been shown to disproportionately expose women to lower-paying opportunities \cite{lambrecht2019algorithmic}, while crime risk assessment tools significantly disadvantage African-Americans \cite{dressel2018accuracy}.


Auction mechanisms, a common method for allocating resources like properties, advertisements, and car licenses, similarly manifest fairness concerns. Notably, car license auctions in several large Asian cities, designed to control traffic and car ownership, often set high bid prices that favor wealthier individuals \cite{tan2020evaluating}. This system bars lower-income residents from owning cars, restricting their mobility and exacerbating socioeconomic divides. Car ownership, often essential for accessing better employment and services, becomes a privilege of the affluent, reinforcing existing social inequities \cite{chen2013bidding}. Similarly, housing auctions intensify economic disparities. Properties auctioned to the highest bidder enable wealthier individuals to secure homes in desirable areas, which are likely to appreciate in value. This not only widens the wealth gap through significant capital gains but also contributes to economic segregation and reduces affordable housing availability, pricing out lower-income groups from prosperous locales \cite{eggum2021housing}.


While auctions are an effective means of resource allocation and revenue generation, our research aims to balance efficiency with fairness in auction design, thus promoting social justice and equality. We first examine the concept of envy-freeness (EF) \cite{gal2016fairest,gamov1958puzzle,todo2011generalizing}, a well-discussed notion in mechanism design, which deems an auction envy-free if no buyer would benefit from swapping her allocation and payment with another. 
However, EF does \emph{not} necessarily address broader concerns of financial accessibility and is inadequate for our purposes. For instance, an auction that allocates the item to the buyer with the highest valuation while demanding a price that exceeds the affordability of other buyers is EF, but favors those with greater financial resources. Therefore, we need a new concept of fairness for auctions to ensure equity among all demographic groups. Our study addresses three main questions: 

\begin{itemize}
    \item[] {\bf (Q1)} How to define fairness for auctions to ensure equity among all demographic groups? 

    \item[] {\bf (Q2)} How to design auctions that are fair and incentive compatible (IC), while maximally ensuring seller's revenue?

    \item[] {\bf (Q3)} How to improve equity at the individual-level as well as the group-level in auction design?
\end{itemize}

\noindent {\bf Contribution.} To answer Q1, we introduce a new fairness notion, {\em $\epsilon$-group fairness}, which aims to minimise the welfare disparity between groups within a defined threshold, $\epsilon$. The parameter $\epsilon$ controls the level of fairness to be achieved by the auction mechanism. See Section~\ref{sec:formuation}. To answer Q2 and Q3, we  present two auction mechanisms: the Group Probability Mechanism (GPM) and the Group Score Mechanism (GSM). GPM assigns winning probabilities across groups, ensuring both incentive compatibility and $\epsilon$-group fairness; see Section~\ref{sec:probability}. GSM, meanwhile, allocates scores to buyers based on group characteristics. The scoring function is optimised via a neural network, trained to improve fairness both at the individual level; See Section~\ref{sec:score}. 
We conduct extensive experiments to empirically investigate the balance between group-level welfare, individual welfare and seller's revenue. We compare our GPM and GSM with the second-price mechanism, that achieves the optimal social welfare, and the simple mechanism, that achieves $0$-group fairness, under synthetic data. The results show that the GPM achieves good group fairness without compromising much of social welfare and of revenue, while GSM successfully achieves good fairness at both group- and individual-levels.
See Section~\ref{sec:experiments}. 


\section{Related work}

\paragraph{Fairness in mechanism design.} Fairness has emerged as a crucial consideration in the field of mechanism design, particularly in the context of resource allocation. In these problems, a prominent objective is to distribute resources among individuals or groups without monetary transactions, a setup often referred to as {\em mechanism design without money}. The fairness notions most frequently studied in this setting are envy-freeness (EF) \cite{gamov1958puzzle} and proportionality (PROP) \cite{steinhaus1949division}. EF, as defined by Gamov and Steinhaus, ensures that no individual envies another after the allocation, meaning no one would prefer another's allocation to their own. PROP, on the other hand, ensures that each individual receives a share proportional to their entitlement from the total available resources.


However, EF and PROP do not always lead to feasible allocations when dealing with individual items, prompting researchers to propose several relaxed fairness notions. These include envy-freeness up to one good (EF1) and envy-freeness for any good (EFX), which permit some level of envy but under strict conditions. EF1 \cite{lipton2004approximately,budish2011combinatorial}, for example, allows envy provided it can be resolved by reallocating a single item, while EFX \cite{caragiannis2019unreasonable} requires that envy can be resolved by reallocating any item in question. The maximin share guarantee (MMS) \cite{budish2011combinatorial} is another relaxed criterion ensuring that when dividing goods into bundles, each participant receives at least the worst one.

%

Recent literature also explores group-level fairness notions, such as group envy-freeness (GEF) \cite{conitzer2019group,benabbou2019fairness}, which ensures that no group envies another after allocations that could be Pareto improved. These work allow for arbitrary group partitioning, reflecting real-world scenarios where groups are not predefined. This flexibility, however, introduces challenges in achieving fairness, especially for indivisible goods, leading to the introduction of group envy-freeness up to one good (GEF1) \cite{conitzer2019group,aziz2021almost}.


Building on these foundational ideas, some studies have extended the notion of envy-freeness to scenarios involving money. 
This line of research is crucial for applications such as rent division \cite{gal2016fairest} and auction design, where monetary considerations significantly impact fairness. For instance, the fairest rent division extends EF to situations where tenants have varying budgets \cite{airiau2023fair}, accommodating real-life financial constraints. In auctions, the concept has been expanded to include envy-freeness between individuals and groups, as well as among groups \cite{todo2011generalizing}, through mechanisms like the Average-Max-Minimal-Bundle, which aims to ensure fairness at multiple levels. Also, \cite{birmpas2022fair} extends the concepts of EF1 and EFX to the group level in the context of sponsored search auctions. 


\paragraph{Fairness in AI.} In the field of AI, fairness are also defined at both individual and group levels. Fairness at individual level requires to treat similar individuals in a similar way \cite{dwork2012fairness}. Specifically, given a similarity metric, fairness-through-awareness \cite{dwork2012fairness} requires individuals who are close under the similarity metric receive similar outcome. Fairness-through-unawareness \cite{grgic2016case}  requires that any protected attributes are not explicitly used in deciding outcomes. Counterfactual fairness\cite{kusner2017counterfactual} requires that the outcome is same in both the actual world and a counterfactual world where the individual changes only her protected attribute. Fairness at group level aims to equalise each group's opportunity to positive outcomes \cite{hardt2016equality}.  Demographic parity  (as referred to as statistical parity) \cite{dwork2012fairness} requires that the demographics of individuals who receive positive (or negative) outcome are consistent with those of the whole population.  Equal odds \cite{hardt2016equality} ensures that the probabilities of an individual in a positive class being correctly assigned a positive outcome and of an individual in a negative class being incorrectly assigned to a positive outcome should be equal for individuals  in protected and unprotected groups. Equal opportunity \cite{hardt2016equality} relaxes the notion of equal odds and ensures only the probability of an individual in a positive class being correctly assigned a positive outcome should be equal for individuals  in protected and unprotected groups.

Most of the studies in mechanism design consider mechanism without money and the corresponding theoretical results and mechanisms can not be applied to our problem. \cite{todo2011generalizing} is the only one that focuses on fairness at group level in auctions. However, their setup is different from ours as we assume that the agents are partitioned into groups by their characteristics while they allow any partition. Also, as we discussed in the introduction, envy-free allocations does not ensure equity among different socioeconomic groups. 
In contrast, fairness in the field of AI is often defined based on pre-defined socioeconomic groups. Also, it aims to equally assign the favourable outcomes to all groups. Therefore, there is a need for new fairness notions by adapting the fairness notion in AI to the auction setting, and design new mechanisms that meet the proposed fairness.

\section{Problem formulation} \label{sec:formuation}
\subsection{Auction setup}
We investigate an auction setting which consists of a single seller, denoted by $s$, and a set of $n$ buyers, denoted by $N$. The seller has an indivisible item for sale. The buyers in $N$ are organised into $m$ mutually disjoint subsets, $N_1,\ldots,N_m$, i.e.,  $N_1 \cup \ldots \cup N_m = N$, $N_1 \cap \ldots \cap N_m = \emptyset$. Intuitively, these sets represent mutually-exclusive socioeconomic groups defined by, e.g., gender, race, income brackets, etc. Let $[m]$ denote the set of natural numbers $\{1,\ldots,m\}$. Each buyer $i \in N$ has a positive real valuation $\theta_i \in [\underline{\theta},\overline{\theta}] \in \mathbb{R}^+$ to the item, where $[\underline{\theta},\overline{\theta}], 0\leq \underline{\theta} < \overline{\theta}$, is the support of $\theta_i$ and it is publicly known. 
Each buyer expresses her valuation by a {\em bid} $\theta_i'$, which is not necessarily the same as the true valuation $\theta_i$. Let $\theta \coloneqq (\theta_i)_{i \in N}$ and $\theta' \coloneqq (\theta_i')_{i \in N}$ denote the true valuations and the reported bids of all buyers, respectively. Using the bids $\theta'$ from all buyers, the auction mechanism decides which buyer wins the item and the price they must pay. The seller values the item at zero, implying no personal loss or gain from parting with the item other than the revenue generated from the sale. 
Let $\Theta \coloneqq [\underline{\theta},\overline{\theta}]^n$ denote the set of all possible valuations of all buyers. We formally define a deterministic mechanism and a randomised mechanism as follows. 

\begin{definition}
\label{def:mechanism}
A {\em deterministic mechanism} $M$ consists of two functions $(\pi,p)$, where $\pi\colon \Theta \to \{0,1\}^n$ is an {\em allocation function} and $p\colon \Theta \to \R^n$ is a {\em payment function}.  
\end{definition}

\noindent A deterministic mechanism $M$ takes the bids $\theta'$ of all buyers as inputs, and outputs an auction outcome $(\pi(\theta'),p(\theta'))$, where $\pi(\theta')\in \{0,1\}^n$ is an {\em allocation result} indicating whom the item is allocated to, while $p(\theta')\in \R^n$ is a {\em payment result} showing how much each buyer would pay. We write $\pi(\theta')$ as $(\pi_1(\theta'),\ldots,\pi_n(\theta'))$ and $p(\theta')$ as $(p_1(\theta'),\ldots,p_n(\theta'))$,  where each $\pi_i(\theta')\in\{0,1\}$ and $p_i(\theta')\in \R$ is buyer $i$'s allocation and payment. Given a mechanism $M$ with bids $\theta'$, the seller $s$ gains the utility $u_s\coloneqq\sum_{i\in N} p_i$, while each buyer $i$ gains a {\em utility} $u_i(\theta')\coloneqq \theta_i \pi_i - p_i$. The {\em social welfare} of $M$ on $\theta'$ is defined as the sum of the utilities of all buyers and the seller, i.e., $sw(\theta')\coloneqq \sum_{i\in N \cup \{s\}} u_i = \sum_{i \in N} \pi_i \theta_i $. The {\em revenue} of $M$ on $\theta'$ is defined as the total payment of all buyers, i.e., $rv(\theta') \coloneqq \sum_{i\in N} p_i$. 

While deterministic mechanisms produce specific, fixed outcomes based on the bids (i.e., one buyer definitively wins the item, and others do not), randomised mechanisms yield probabilities that describe expected outcomes over multiple realisations of the mechanism, offering a distribution of possible outcomes.

\begin{definition}
A {\em randomised mechanism} $M$ is one that, given $\theta' \in \Theta$, outputs $(\pi, p)$ such that $\pi$ and $p$ are randomised allocation and payment functions, respectively.
\end{definition}

\noindent In a randomised mechanism $M$ over $\theta'$, we denote by $\Pi_i(\theta')\in [0,1]$ and $P_i(\theta')\in \R$ the expected allocation and expected payment to buyer $i$, resp. 
Buyer $i$'s expected utility, denoted by $U_i(\theta')$, is $\theta_i\Pi_i(\theta')-P_i(\theta_i)$. Seller $s$'s expected utility, denoted by $U_s(\theta')$, is $\sum_{i\in N} P_i(\theta_i)$.
We also define the expected social welfare and expected revenue of $M$ over $\theta'$, denoted by $SW(\theta')$ and $RV(\theta')$, resp. 

Note that a deterministic mechanism can be treated as a special case of randomised mechanism such that the mechanism returns an allocation and a payment function  with probability $1$.
From now on, we refer to a randomised mechanism simply as a mechanism. 
Next we define several desirable properties of a mechanism. Let $\theta_{-i} \coloneqq (\theta_j)_{j \in N \backslash \{i\}}$ and $\Theta_{-i}$ be the set of all possible $\theta_{-i}$. 
Intuitively, the buyers should be incentivised to truthfully report their valuations as doing so would give her the highest and non-negative expected utility: 

\begin{itemize}[leftmargin=*]
    \item A mechanism $M$ is {\em incentive compatible} (IC) if for all $i\in N$, all $\theta_i, \theta_i' \in [\underline{\theta}, \overline{\theta}]$ and for all $\theta_{-i} \in \Theta_{-i}$, we have $U_i(\theta_i,\theta_{-i}) \geq U_i(\theta_i',\theta_{-i}).$ 
    \item  A mechanism $M$ is {\em individual rational} (IR) if for all $i \in N$ and all $\theta_{-i}\in \Theta_{-i}$, we have $U_i(\theta_i,\theta_{-i})\geq 0$. 
\end{itemize}

\subsection{Group fairness} 

In auction settings where the seller distributes items of social significance, incorporating fairness among different buyer groups becomes critical. This is particularly important when certain groups might have inherent advantages over others in winning the auction. To ensure equitable treatment across all groups, we introduce the concept of group fairness below.

For each group $N_k$, where $k\in [m]$, we define the group's social welfare $sw_k(\theta')$ under a deterministic mechanism $M=(\pi,p)$ over bids $\theta'$ by the sum of the valuations of buyers within $N_k$, i.e.,
$sw_k(\theta') \coloneqq \sum_{i\in N_k} \pi_i \theta_i.$
%
To assess a mechanism's fairness, we consider the expected social welfare $SW_k(\theta')$ for any group $N_k$. By focusing on achieving comparable expected social welfare across all groups, the mechanism promotes a fair distribution of benefits, crucial for items that serve the community or public interest.

\begin{definition}\label{def:fairness}
Given a parameter $\epsilon \geq 0$, a mechanism $M$ satisfies {\em $\epsilon$-group fairness} if for all $\theta\in \Theta$ 
we have 
\begin{equation}\label{eqn:EGF}
\max_{k,\ell \in [m]} |SW_k(\theta) - SW_\ell(\theta)| \leq \epsilon.  
\end{equation}
\end{definition}

\noindent The parameter $\epsilon$ regulates the extent of fairness deviation. A smaller value of $\epsilon$ means better fairness between groups in the mechanism. 


\medskip 
We aim to design an auction mechanism that is IC and IR, while satisfying $\epsilon$-group fairness, for reasonable $\epsilon$. 
Designing such an mechanism is not trivial, which is illustrated by a straightforward adaption of the second price auction mechanism such that it achieves group fairness. We call this mechanism as {\em simple mechanism} and it works as follows. The simple mechanism assigns a proper probability $\Pr_k \in [0,1] \in \R^+$ to each group $N_k, k\in [m]$, and selects one group according to the probability distribution.  Then it runs a second price auction in the selected group. By setting the probability of each group be  proportional to the top valuation in other groups would ensure group fairness. 
However, it turns out that this mechanism is not IC, which is shown in the following example.

\begin{example}
We consider the auction with two groups $A$ and $B$, where Group A consists of buyers $a, b$ and $c$ with valuations $9,8$ and $7$, while Group $B$ consists of $d, c$ and $e$ with valuations $7,3$ and $2$. 
We first consider the case where all buyers bids truthfully. 
In the second price auction, the item is allocated to the buyer with the highest bids within the group, i.e., either $a$ in $A$ or $d$ in $B$. We consider $\epsilon=0$. To ensure group fairness, by Equation~\eqref{eqn:EGF}, we have $v_a \Pr_A = v_d \Pr_B$. Also, we know $\Pr_A + \Pr_B=1$ and $\Pr_A \geq 0, \Pr_B \geq 0$. Then we have $\Pr_A=7/16, \Pr_B=9/16$, which defines the probability distribution in the first step. 
Under such distribution, the expected utility of buyer $d$ is $9/4$. 

Next, we assume that Buyer $d$ bids $6$ instead and the other buyers bid truthfully. Again to ensure group fairness, we have the distribution $\Pr_A=6/15, \Pr_B=9/15$. Under this distribution, the expected utility of $d$ increases to $35/15$. In other words, Buyer $d$ benefits from misreporting and thus the simple mechanism is not IC.
\end{example}


\section{Group probability mechanism}
\label{sec:probability}

Notice that the simple mechanism fails to achieve incentive compatibility because the expected allocation is increasing in the bid and the buyers have incentive to lower their bids to get higher chance to win. To circumvent this issue, we make the expected allocation of each buyer independent from her bid, which is the main idea of our first mechanism, {\em group probability mechanism} (GPM). 

\subsection{Mechanism}
We now describe GPM in detail: Given the bids of all buyers $\theta'$ and a parameter $\epsilon_G>0$, the mechanism works as follows. 

\begin{enumerate}[leftmargin=*]
    \item {\bf Bidder Partition.} The buyers are randomly assigned to exactly one of the two disjoint sets: $\Stat$ with probability $\frac{1}{2}$ and $\SecPrice$ with probability $\frac{1}{2}$. The buyers in $\SecPrice$ will participate in a second price auction while the buyers in $\Stat$ will be used to calculate the winning probability of each group (see next step). 
    \item{\bf Calculation of group winning probabilities.} This step is to use the buyers in $\Stat$ to calculate the probability $\Pr_k, k \in [m]$ that group $N_k$ wins the item. 
    Within the buyers in  set $\Stat$, assume a second price auction is run among the buyers in each group. 
    Let $w_k$ be the winner in the second price auction for group $N_k$. The probabilities $\Pr_k$ are obtained by solving Problem~\eqref{eqn:opt}, as shown below.

    \begin{subequations}
    \label{eqn:opt}
    \begin{align}
    \max_{\Pr_1,\ldots,\Pr_m} & \sum_{k=1}^{m}  p_{w_k} \Pr_k \label{obj} \\
    \text{s.t. } &  \max_{k, l \in [m]}|\Pr_k \theta_{w_k}' - \Pr_l \theta_{w_l}'| \leq \epsilon \label{con:EGF} \\
      & \sum_{k=1}^m \Pr_k  = 1 \label{con:sum}\\
      & 0 \leq \Pr_k \leq 1, \forall k \in [m] \label{eqn:nonneg}
    \end{align}    
    \end{subequations}
    The Objective~\eqref{obj} is to maximise the expected revenue while the Constraint~\eqref{con:EGF} ensures $\epsilon$-group fairness.
    \item {\bf Second price auction.} This step is to select a winner among the buyers in set $\SecPrice$ and to set the price for the winner. 
    \begin{enumerate}
        \item {\bf Group selection.} The group probability $\Pr_k, k \in [m]$ obtained from the buyers in $\Stat$ is assigned to group $N_k$ in $\SecPrice$. Then a group $N_{k^*}$ is selected randomly according to the assigned group winning probabilities. 
        \item {\bf Individual selection.} A second price auction is conducted among the buyers in group $N_{k^*}$. That is, the buyer $w$ with the highest bid in $N_{k^*}$ is allocated the item and the price is set to be the second highest bid of group $N_{k^*}$ in set $\SecPrice$, i.e., $\pi_w=1, p_w=\theta_2'$, where $\theta_2'$ is the second highest bid.    
    \end{enumerate}  
\end{enumerate}

The whole process is shown in Algorithm~\ref{alg:GPM}. 

\begin{algorithm}[h]
\footnotesize
\caption{Group Probability Mechanism}
\label{alg:GPM}
    \begin{algorithmic}[1]
    \Require bids of all buyers $\theta'$ and group fairness parameter $\epsilon_G$
    \Ensure allocation result $\pi(\theta')$ and payment result $p(\theta')$
    \State initialise  $\Stat = \emptyset$, $\SecPrice = \emptyset$, $\pi(\theta')=\mathbf{0}$, and $p(\theta')= \mathbf{0}$  
    \State allocate each buyer into $\Stat$ with prob. $\frac{1}{2}$ and $\SecPrice$ with prob. $\frac{1}{2}$ 
    \State calculate the group winning prob. $\Pr_k$ for each group $k \in [m]$ by solving Problem~\eqref{eqn:opt} over $\Stat$
    \State assign the group winning prob. $\Pr_k$ to each group $N_k$ in $\SecPrice$
    \State select a group $N_{k^*}$ in $\SecPrice$ according to the group winning prob. distri.
    \State run a second price auction in group $N_{k^*}$ and set $\pi_{w^*}=1$ and $p_{w^*}=\theta_2'$ for winner $w^*$ 
    \State \Return $\pi(\theta')$ and $p(\theta')$
    \end{algorithmic}
\end{algorithm}

Now we show how we solve Problem~\eqref{eqn:opt} over set $\Stat$.  We use a general method for solving optimisation programs, the Lagrange multiplier method. Let $\lambda_{kl}, k\neq l, k,l\in [m]$, $\mu$ and $\gamma_{k}, \tau_{k}, k \in [m]$ be the Lagrange multipliers. The Lagrange function is 
\begin{align*}
& L(\Pr_1,\ldots, \Pr_k, \lambda_{12},\ldots, \lambda_{m-1,m}, \mu, \gamma_{1},\ldots,\gamma_{k},\tau_{1},\ldots,\tau_{k}) \\
=& \sum_{k\in [m]} p_{w_k} \Pr_{k}  
- \sum_{k,l \in [m]} \lambda_{k,l}\left(|\Pr_k \theta_{w_k}' - \Pr_l \theta_{w_l}'| - \epsilon\right) \\
& - \mu \left(\sum_{k\in [m]} \Pr_k  - 1\right) 
- \sum_{k\in [m]}\gamma_{k} (- \Pr_k) - \sum_{k\in [m]}\tau_{k} (\Pr_k-1)
\end{align*}
Then we get the solution to Problem~\eqref{eqn:opt} by solving the following system. 
\begin{align*}
& \partial  L/ \partial \Pr_k=0, k \in [m]\\
&|\Pr_k \theta_{w_k}' - \Pr_l \theta_{w_l}'| - \epsilon = 0, \forall 1\leq k,l \leq m \\
&\sum_{k=1}^m \Pr_k  - 1 =0 \\
& \gamma_{k} (- \Pr_k) =0, \forall 1\leq k \leq m \\
& \tau_{k} (\Pr_k-1) =0, \forall 1\leq k \leq m  \\
& \lambda_{12},\ldots, \lambda_{m-1,m}, \mu, \gamma_{1},\ldots,\gamma_{k},\tau_{1},\ldots,\tau_{k} \geq 0
\end{align*}
\subsection{Analysis}
Next we show that GPM have several desirable properties. 
Before that, we present a theorem stating the sufficient conditions for a mechanism being IC and IR. 

\begin{theorem}[\cite{archer2001truthfulMF}]
\label{thm:ICIR}
A mechanism $M(\pi,p)$ is incentive compatible and individully rational if and only if for every $i\in N$, 
\begin{enumerate}
    \item $\Pi_i(\theta')$ is monotonically non-decreasing in $\theta_i'$, 
    \item $\int_{\underline{\theta}}^{\overline{\theta}} \Pi(x,\theta_{-i}')dx <  \infty$ for all $\theta_{-i}'$, and
    \item $P_i(\theta')=\theta_i' \Pi_i(\theta') -\int_{\underline{\theta}}^{\theta_{i}'}\Pi_i(x,\theta_{-i}')dx$.  \qed
\end{enumerate}
\end{theorem}

Then we prove that GPM is IC and IR by showing that it satisfies the conditions listed in Theorem~\ref{thm:ICIR}. 

\begin{lemma}
\label{lem:GPIC&IR}
The group probability mechanism is incentive compatible and individually rational. 
\end{lemma}

\begin{proof}
Consider an arbitrary buyer $i$ associated with valuation $\theta_i$ and group $N_k$. 
It is straightforward to see that for buyer $i$ the expected allocation $\Pi_i$ is monotonically non-decreasing in her bid $\theta_i'$. Firstly, whether $i$ is allocated to set $\SecPrice$ or $\Stat$ in Step 1 and whether the group $N_k$ is selected in Step 3(a) are independent from the $i$'s bid $\theta_i'$. Also, $\Pi_i$ is increasing in her bid in the second price auction of Step 3(b). 

Next we show that the expected payment of buyer $i$ is exactly in the form of Theorem~\ref{thm:ICIR}. 
By the payment rule of our mechanism, we have the expected payment as $P_i(\theta') = \frac{1}{2} \Pr_{k} \theta_2'$. 
Then by Thm.~\ref{thm:ICIR}, the expected payment should be 
$P_i(\theta')  =  \theta_i' \Pi_i(\theta_i') - \int_{\overline{\theta}}^{\theta_i'} \Pi_i(x) d x  
= \theta_i'  \Pi_i(\theta_i') - \left(\int_{\theta_2'}^{\theta_i'} \Pi_i(x) d x + \int_{\underline{\theta}}^{\theta_2'} 0 d x \right) 
=  \theta_i' \times \frac{1}{2} \Pr_k - \frac{1}{2} \Pr_k (\theta_i' - \theta_2') = \frac{1}{2} \Pr_k \theta_2'$.

Lastly, it is easy to see the integral on $\Pi_i$ is finite. 
\end{proof}

\begin{lemma}
\label{lem:GPGF}
The group probability mechanism asymptotically satisfies $\epsilon$-group fairness when the number of buyers $n \to \infty$. 
\end{lemma}

\begin{proof}
W.l.o.g., we assume that Groups $k,l$ have the maximum difference in expected social welfare. Let $\mathcal{D}_k$ and $\mathcal{D}_l$ denote the distributions of the valuations in Groups $k$ and $l$, resp. Let $\theta_{N_k}$ and $\theta_{N_l}$ be the valuations of buyers in Groups $k$ and $l$, resp.   
Then we let $\mu_k \coloneqq \int_{\theta_{N_k} \sim \mathcal{D}_k} \int_{\Pi_i\sim M} \Pi_i \theta_i$ and $\mu_l \coloneqq \int_{\theta_{N_l} \sim \mathcal{D}_l} \int_{\Pi_i\sim M} \Pi_i \theta_i$. 

By law of large numbers, we have $\Pr[\lim_{n\to \infty} SW_k = \mu_k]=1$ and $\Pr[\lim_{n\to \infty} SW_l = \mu_l]=1$ for both sets $\Stat$ and $\SecPrice$. Then we have 
\begin{align*}
&\Pr\left[\lim_{n\to \infty} |SW_k (\theta'_{\SecPrice}) - SW_l (\theta'_{\SecPrice})|= |\mu_k - \mu_l|\right]  = 1, \\
&\Pr\left[\lim_{n\to \infty} |SW_k (\theta'_{\Stat}) - SW_l (\theta'_{\Stat})| = |\mu_k - \mu_l|\right] =  1.   
\end{align*}

Also, as the probabilities are derived from Problem~\eqref{eqn:opt}, we have  
$|SW_k (\theta'_{\Stat}) - SW_l (\theta'_{\Stat})| = |\mu_k - \mu_l| = \epsilon.$
Then we have 
$\Pr\left[\lim_{n\to \infty} |SW_k (\theta'_{\SecPrice}) - SW_l (\theta'_{\SecPrice})|= \epsilon \right]  = 1.$
\end{proof}

\section{Group score mechanism}
\label{sec:score}

Notice that even though GPM ensures IC, IR and group fairness, it is not fair at individual level. That is, the majority of the buyers (except for the buyer with the highest value in each group) has the expected utility of $0$. To alleviate the unfairness between individuals, we propose our second mechanism, named {\em group score mechanism} (GSM).
Different from GPM that assigns a probability to each group and then conducts an auction within the selected group, leading a zero winning probability for most of the buyers, GSM assigns a non-zero probability to all buyers. In this way, the difference in the expected utilities between the winner and other buyers is narrowed down and thus the individual unfairness is alleviated.

The key in the design of GSM is that it calibrates the expected social welfare of each group by a {\em score function}. A score function, $\sigma: [\underline{\theta},\overline{\theta}] \to \R^+$, is non-decreasing in bid. We design a score function for each group, which is shared by the buyers in the group. The underlying idea is that the score function aligns the diverse valuation ranges of different groups into comparable scales, thereby equalising the winning chances for both disadvantaged and advantaged groups.

\subsection{Mechanism}

\begin{algorithm}[t]
\caption{Group Score Mechanism}
\footnotesize
\label{alg:GSM}
    \begin{algorithmic}[1]
    \Require bids of all buyers $\theta'$ and group fairness parameter $\epsilon_G$ 
    \Ensure allocation result $\pi(\theta')$ and payment result $p(\theta')$
    \State initialise  $\Stat = \emptyset$, $\SecPrice = \emptyset$, $\pi(\theta')=\mathbf{0}$, and $p(\theta')= \mathbf{0}$  
    \State allocate each buyer into $\Stat$ with prob. $\frac{1}{2}$ and $\Auction$ with prob. $\frac{1}{2}$ 
    \State learn the score functions $\sigma_k$ for each group $k \in [m]$ over $\Stat$
    \State calculate the individual winning prob. $\Pi_i$ of each buyer $i\in N$ by Equation~\eqref{eqn:Pi} 
    \State select a winner $i$ according to the individual winning prob. distr.
    \State set the payment of the winner $i$ by Equation~\eqref{eqn:p}
    \State \Return $\pi(\theta')$ and $p(\theta')$
    \end{algorithmic}
\end{algorithm}

GSM takes the bids of all buyers $\theta'$, and a group fairness parameter $\epsilon$ 
and returns the allocation and payment results of the buyers. 

\begin{enumerate}[leftmargin=*]
\item {\bf Bidder Partition.} The Buyers are randomly allocated into exactly one of the two disjoint sets: $\Stat$ with probability $\frac{1}{2}$ and $\Auction$ with probability $\frac{1}{2}$. The buyers in $\Auction$ will participant in an auction while the buyers in $\Stat$ will be used to calculate the score functions. 

\item{\bf Calculation of group score functions.} This step is to use the buyers in $\Stat$ to calculate the score functions $\sigma_k$ for each group $k \in [m]$. Assume an auction in Step (3) is conducted within the buyers in set $\Stat$. The optimal group score functions can be obtained by solving Problem~\eqref{eqn:opt2} below.

\begin{subequations}
    \label{eqn:opt2}
    \begin{align}
    \max_{\sigma_1,\ldots,\sigma_m} & \sum_{i\in N} p_i \Pi_i \label{obj2} \\
    \text{s.t. } &  \max_{k, l \in [m]}|\sum_{i\in N_k} v_i \Pi_i- \sum_{i\in N_l} v_i \Pi_i | \leq \epsilon \label{con:EGF2} \\
      & \Pi_i(\theta_i',\theta_{-i}') \leq \Pi_i(\theta_i'',\theta_{-i}'), \forall \theta_i'\leq \theta_i'' \label{con:inc2} \\
      & \sum_{i\in \Auction} \Pi_i  = 1 \label{con:sum2}\\
      & 0 \leq \Pi_i  \leq 1, \forall i \in N, \forall \theta_i'\in [\underline{\theta},\overline{\theta}] \label{eqn:nonneg2} \\
      & 0 \leq \int_{\underline{\theta}}^{\theta_i'} \Pi_i(x,\theta_{-i}) dx \leq +\infty, \forall i\in N \label{eqn:finite}
    \end{align}    
\end{subequations}

The Objective~\eqref{obj2} is also to maximise the expected revenue. 
The problem is hard to solve, and thus we deploy the dual ascent algorithm \cite{boyd2011distributed} to get the solution. We present the details in Sec.~\ref{sec:learning}.
 
\item{\bf Auction.} This step is to select a winner among the buyers in $\Stat$ and to determine the payment of the winner. 
\begin{enumerate}
    \item {\bf Calculation of individual winning probabilities.} In this step, the mechanism assigns each buyer $i$ a probability $\Pi_i$ that defines the chance buyer $i$ wins the item. Given the score functions $\sigma_1,\ldots,\sigma_m$, the winning probability of $i$ in group $N_k$ is
    \begin{equation}
    \label{eqn:Pi}
    \Pi_i(\theta')= \frac{\sigma_k(\theta_i')}{\sum_{l=1}^m \sum_{j\in N_l} \sigma_l (\theta_j') }.
    \end{equation}
    \item {\bf Winner selection and payment setting. } Given the obtained probability distribution, we randomly select a buyer $i$ as the winner. The payment of $i$ is set to be 
    \begin{equation}
    \label{eqn:p}
    p_i = \theta_i' - \int_{\underline{\theta}}^{\theta_i'} \Pi_i(x,\theta_{-i}') dx / \Pi_i(\theta').
    \end{equation}
\end{enumerate}
\end{enumerate}

Algorithm~\ref{alg:GSM} shows the whole process of GSM.

\subsection{Analysis}

\begin{lemma}
\label{lem:GSIC&IR}
    The group score mechanism is incentive compatible and individually rational. 
\end{lemma}

\begin{proof}
The IC and IR properties of GSM can also be proved by showing that the expected allocation and expected payment satisfies the conditions in Theorem~\ref{thm:ICIR}. We first show that the expected allocation is non-decreasing. Consider an arbitrary buyer $i\in N$ in group $N_k$. The group score function $\sigma_k$ is determined by her group $k$ and the bids of set $\Stat$, and thus independent from her bid. Also, the group score function $\sigma_k$ is non-decreasing in her bid $\theta_i'$ be default. Therefore, her  $\Pi_i(\theta')$ in Eqn.~\eqref{eqn:Pi} is non-increasing in $\theta_i'$. 

Then, the expected payment of $i$ is $p_i \times \Pi_i(\theta') = \theta_i' \Pi_i(\theta') - \int_{\underline{\theta}}^{\theta_i'} \Pi_i(x,\theta_{-i}') dx$, which meets in Condition (3) of Theorem~\ref{thm:ICIR}. 
\end{proof}

\begin{lemma}
\label{lem:GSGF}
The group score mechanism asymptotically satisfies $\epsilon$-group fairness when $n \to \infty$. \qed
\end{lemma}

The $\epsilon$-group fairness of GSM can be easily proved following the arguments for Lemma~\ref{lem:GPGF}.


\subsection{Learning score functions}
\label{sec:learning}

Now the key is to find proper group score functions for all groups by solving Problem~\eqref{eqn:opt2}. We apply QMIX \cite{rashid2020monotonic} and neural-network techniques to design parameterised group score functions, and learn its parameters using dual ascent algorithm \cite{boyd2011distributed}. 

\noindent{\bf Design and parmetrisation of the score function.} The group score function $\sigma_k, k \in [m]$ is designed as $\sigma_k = a_k(b_k f(\theta')+c_k) + d_k$, where $a_k,b_k,c_k,d_k$ are parameters for each group $k$ and $f(\theta'): \Theta \to \R_+$ is a monotonically non-decreasing function of $\theta_i'$, e.g., $f$ can be $\theta_i', \theta_i'^2$ or $\exp(\theta_i')$. 
The parameters $a_k,b_k,c_k,d_k$ for each group $k$ are generated  by three separate three-layer neural networks. The neural network is composed of three functions. 
\begin{enumerate}[leftmargin=*]
\item The first layer is a linear function $\ell^{(1)}: \Theta \rightarrow \R^h$ that takes the bids as input. Specifically, it is $\ell^{(1)}(\theta') = A^{(1)}\theta' + \beta^{(1)}$, where $A^{(1)}$ and $\beta^{(1)}$ are the coefficients and $h$ is a constant denoting the number of neural employed in the second layer. 
\item The second layer is a function $\ell^{(2)}: \R^h \rightarrow \R^h$ that takes the output from the first layer. It is a ReLU function, i.e., $\ell^{(2)}(x)= x$, if $x \geq 0$; otherwise, it is $0$. 
\item The third layer is a linear function $\ell^{(3)}: \R^h \rightarrow \R$ that takes the the output from the second layer. Specifically, it is $\ell^{(3)}(x) = A^{(3)} x + \beta^{(3)}$, where $A^{(3)}$ and $\beta^{(3)}$ are coefficients.
\end{enumerate}

\noindent For each group $k$, the parameters $a_k, b_k, d_k$ are calculated by the three layers $|\ell^{(3)}(\ell^{(2)}(\ell^{(1)}(\theta')))|$, while the parameter $c_k$ is computed by a single linear function
$\ell^{(c)}$ where $\ell^{(c)}= \left|A^{(c)}\theta' + \beta^{(c)}\right|$. 

\noindent{\bf Learning the parameters of the score function. }
Let $\mu \coloneqq (a_1, b_1, c_1, d_1, \ldots, a_m, b_m, c_m, d_m)$ be all learnable parameters in the score functions for all groups and $\sigma^{\mu}$ be the group score function parameterised by $\mu$.  We deploy the dual ascent techniques \cite{boyd2011distributed} to learn the parameters $\mu$ for score functions and approximate the optimal solution to Problem~\eqref{eqn:opt2}. 

Due to the design of the score function, Problem ~\eqref{eqn:opt2} can be simplified. Constraint~\eqref{con:inc2} on the non-decreasing of expected allocation is satisfied as the group score function $\sigma^{\mu}$ is designed to be non-decreasing. Constraints~\eqref{con:sum2} and~\eqref{eqn:nonneg2} are met due to the design of the expected allocation in Equation~\eqref{eqn:Pi}.
Constraint~\eqref{eqn:finite} is also met as both the expected allocation $\Pi_i$ and the bids $\theta'$ are positive and finite.
Thus we omit the three constraints and rewrite Problem~\eqref{eqn:opt2} as,
\begin{equation}
    \label{eqn:opt3}
    \begin{aligned}
    \min_{\sigma^{\mu}} & -\sum_{i\in N} p_i \Pi_i \\ 
    \text{s.t. } &  \max_{1 \leq k, l \leq m}|\sum_{i\in N_k} v_i \Pi_i- \sum_{i\in N_l} v_i \Pi_i | < \epsilon 
    \end{aligned}   
\end{equation}

We call Problem~\eqref{eqn:opt3} as the primal problem and establish its dual problem. Let $\psi(\sigma^{\mu}) \coloneqq \sum_{i\in N} p_i\Pi_i$ be the objective, and $g(\sigma^{\mu})\coloneqq \max_{1 \leq k, l \leq m}|\sum_{i\in N_k} v_i \Pi_i- \sum_{i\in N_l} v_i \Pi_i | -\epsilon$ be the difference between the maximum group unfairness and the targeted group fairness. 
Now we can use the Lagrange multiplier method to transform the problem into the dual.  We write the Lagrangian of the primal problem as $L(\sigma^{\mu},\lambda)= \psi(\sigma^{\mu}) +\lambda g(\sigma^{\mu})$, where $\lambda \geq 0$ is the Lagrangian multiplier. The dual problem is as follows: 
\begin{equation}
\begin{aligned}
    \max_{\sigma^{\mu}} & \inf L(\sigma^{\mu},\lambda)\\
    \mathrm{s.t. } & \lambda \geq 0
\end{aligned}  
\end{equation}
The dual problem aims to maximise the lower bound on the solution to the primal problem subject to the nonnegativity constraint on $\lambda$. 

Then, we apply the dual ascent method to approach the dual problem. Let $\alpha$ be the learning rate. We update the Lagrangian multiplier $\lambda$ by maximising $\inf L(\sigma^{\mu},\lambda)$ and update $\mu$ by minimising $L(\sigma^{\mu},\lambda)$ interchangeably by their gradients as the following:
\begin{equation*}
\begin{aligned}
    \lambda=\lambda+\alpha \Delta_{\lambda} L(\sigma^{\mu},\lambda),\\
    \sigma^{\mu}=\arg \min_{\sigma^{\mu}} L(\sigma^{\mu},\lambda).
\end{aligned}
\end{equation*}

\begin{algorithm}[h!]
\caption{Dual ascent for GSM}
\footnotesize
\label{alg:dual}
	\begin{algorithmic}[1]
	\Require the bids $\theta'_{\Stat}$ of buyers in $\Stat$, learning rate $\alpha$, number of episodes $T$, function $f(\cdot)$
	\Ensure group score function $\sigma_k, 1\leq k\leq m$
        \State initialise $t=0$, initialise $\lambda$ and $\mu$ and $\text{Fair}=\text{False}$
        \While {$t\leq T$}
        \State generate group score functions $\sigma^{\mu}$
        \State compute Lagrangian $L(\sigma^{\mu}, \lambda)$
        \State set $\sigma^{\mu*} =\arg\min_{\sigma^{\mu}} L(\sigma^{\mu},\lambda)$
        \If {$\epsilon$-group fairness (Constraint~\eqref{con:EGF2}) is satisfied}
         \State set $\sigma^{\mu}=\sigma^{\mu*}$
         \State set $\text{Fair}=\text{True}$
        \EndIf
        \State set $\lambda=\lambda+\alpha \Delta_{\lambda} L(\sigma^{\mu},\lambda)$ 
        \State set $t=t+1$
        \EndWhile
        \State \Return $\sigma_k, 1\leq k \leq m$ if $\text{Fair}$; No Solution, otherwise
	\end{algorithmic}
\end{algorithm}

The whole process of dual ascent algorithm is shown in Algorithm~\ref{alg:dual}. The algorithm takes the bids $\theta'_{\Stat}$ of buyers in $\Stat$, learning rate $\alpha$, number of episodes $T$, function $f(\cdot)$ as inputs and returns the  group score functions $\sigma_k, k \in [m]$ parameterised by the learned $\mu$. The algorithm first initialises the number of episode $t$, the Lagrangian multiplier $\lambda$ and the parameters $\mu$ of the neural networks, and also sets the status of Fair as ``False''. Then the algorithm updates the Lagrangian multiplier $\lambda$ and parameters $\sigma^{\mu}$ iteratively until the number of episodes reaches $T$. In each iteration, given the neural networks parameterised by updated $\mu$, the algorithm generates the group score functions $\sigma^{\mu}$. Given $\sigma^{\mu}$ and the updated $\lambda$, the algorithm computes the Lagrangian $L(\sigma^{\mu}, \lambda)$. After that, the algorithm finds $\sigma^{\mu}$ that minimises the Lagrangian using stochastic gradient descent (SGD) method. Then the parameters $\mu$ and the corresponding group score functions $\sigma^{\mu}$ are updated if the $\epsilon$-group fairness constraint is met. Otherwise,  $\lambda$ is updated using the gradient of the (updated) Lagrangian, and the algorithm goes to the next iteration. Finally, the algorithm returns the trained group score functions $\sigma_k, k \in [m]$.

\section{Experiments} \label{sec:experiments}
We conduct experiments to validate the performance of our proposed mechanisms. Our experiments serve to answer the following questions: {\bf (1)} Social welfare and revenue are two key performance metrics of a mechanism. We would like to investigate the performance of our mechanisms in terms of these two criteria, when compared with the optimal case, i.e., the second-price auction, which is IC but not fair, and the best group fair case, i.e., the simple mechanism, which is group fair but not IC. {\bf (2)} When ensuring the fairness at group level, the fairness at individual level is another metrics. We would also like to investigate the individual fairness obtained by the mechanisms. {\bf (3)} Also, we would like to verify the robustness of our mechanism in different setups, including {\bf (a)} valuations drawn from different distributions, {\bf (b)} varying group sizes, and {\bf (c)} varying group fairness levels. {\bf (4)} In GSM, function $f(\theta')$ for the group score functions can be any function that is non-decreasing in $\theta_i'$. We would like to test the effect of the function $f(\theta')$ on the performance.

\subsection{Setup}
\paragraph{Group Size.} In our experiment, we focus on the two-group cases, i.e., $N= N_1 \cup N_2$. In real world, groups can exhibit either equal or unequal size. For instance, the groups divided by gender tend to have similar sizes, while the groups divided by ethnicity have significantly different sizes. To capture this, we set the sizes of the two groups as $(|N_1|,|N_2|)= (100, 900), (300, 700),$ and $(500, 500)$. 

\paragraph{Valuation.} We generate four sets of random numbers to represent buyers' valuations. 
(a) Valuations are drawn from uniform distributions: (a.1) $\theta_{N_1}\sim \mathcal{U}(0,10)$ while $\theta_{N_2}\sim \mathcal{U}(0,8)$, and (a.2) $\theta_{N_1}\sim \mathcal{U}(0,10)$ while $\theta_{N_2}\sim \mathcal{U}(0,4)$. (b) Valuations are drawn from normal distributions: (b.1) $\theta_{N_1} \sim \mathcal{N}(5,1)$ while $\theta_{N_2}\sim \mathcal{N}(4,1)$, and (b.2) $\theta_{N_1} \sim \mathcal{N}(5,1)$ while $\theta_{N_2}\sim \mathcal{N}(2,1)$.
The differences in average valuation between the two groups in (a.1) and (b.1) are both $20\%$, while those in (a.2) and (b.2) are both $60\%$.

\paragraph{Mechanisms.} In addition to our mechanism, GPM and GSM, we also run two other mechanisms, second price auction mechanism and the simple mechanism (introduced in Sec.~\ref{sec:formuation}) as the benchmarks. The second price auction ensures IC but has no fairness guarantee while the simple mechanism is $0$-group fair but fails to ensure IC. The former provides an upper bounds on social welfare and revenue that all IC mechanisms could obtain. The latter provides an upper bounds on the group fairness.
When implementing the simple mechanism, we assume that the buyers truthfully report their valuations. 

\paragraph{Score function.} In GSM, we consider different formats of $f(\theta')$ in the design of the group score functions. Specifically, we set $f(\theta')$ as four formats: linear function $\theta_i'$, log function $\log(\theta_i'+1)$, square function $\theta_i'^2$ and exponential function $\exp{\theta_i'}$.  

\paragraph{Group fairness.} We consider different levels of group fairness, where $\epsilon = \{0.5, 0.75, 1.0, 1.25, 1.5\}$.

\paragraph{Evaluation metrics.} We evaluate the performances of mechanisms in terms of revenue, social welfare and individual fairness. The revenue and social welfare are defined as in Section~\ref{sec:formuation}. The individual fairness level, denoted by $\epsilon^I \in \R^+$, is defined as the maximum difference in expected gain of individual buyers within each group, i.e., $\epsilon^I \coloneqq \max_{k\in [m]} \max_{i,j\in N_k} |\theta_i \Pi_i - \theta_j \Pi_j|$.  

\paragraph{Implementation.} For each type of the distributions, we generate three sets of valuations with different size differences. For each valuation set, we vary group fairness levels.
For each setup, we run $100$ times and return the average values. 
GSM is implemented in Python 3.9 on NVIDIA GeForce RTX 3090 Ti GPU. All mechanisms are implemented in Python 3.9 on Apple M1 Pro CPU. 
The code is available on \url{https://anonymous.4open.science/r/fairness-EB4C/}.


\begin{figure}[t]
    \centering
    \label{fig:results-uniform}
    \includegraphics[width=\columnwidth]{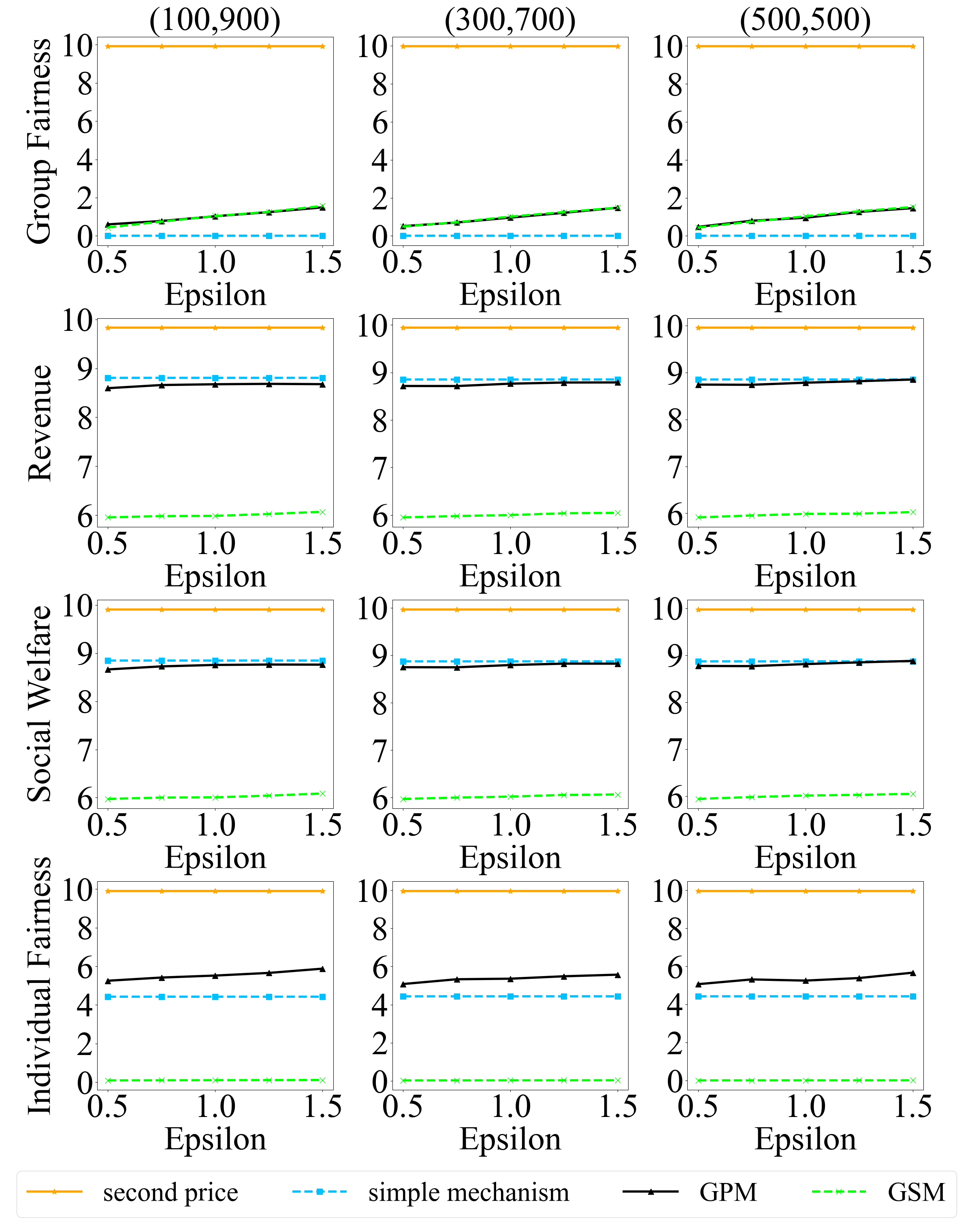}
    \caption{Performance of GPM, GSM with linear function, second-price auction and simply mechanism under valuations drawn from $\mathcal{U}(0,10)$ and $\mathcal{U}(0,8)$ and group sizes $(100,900)$ (left), $(300,700)$ (middle), and $(500,500)$ (right). }
    \vspace{0.5cm}
\end{figure}

\begin{figure}[t]
    \centering
    \label{fig:results-function} 
    \includegraphics[width=\columnwidth]{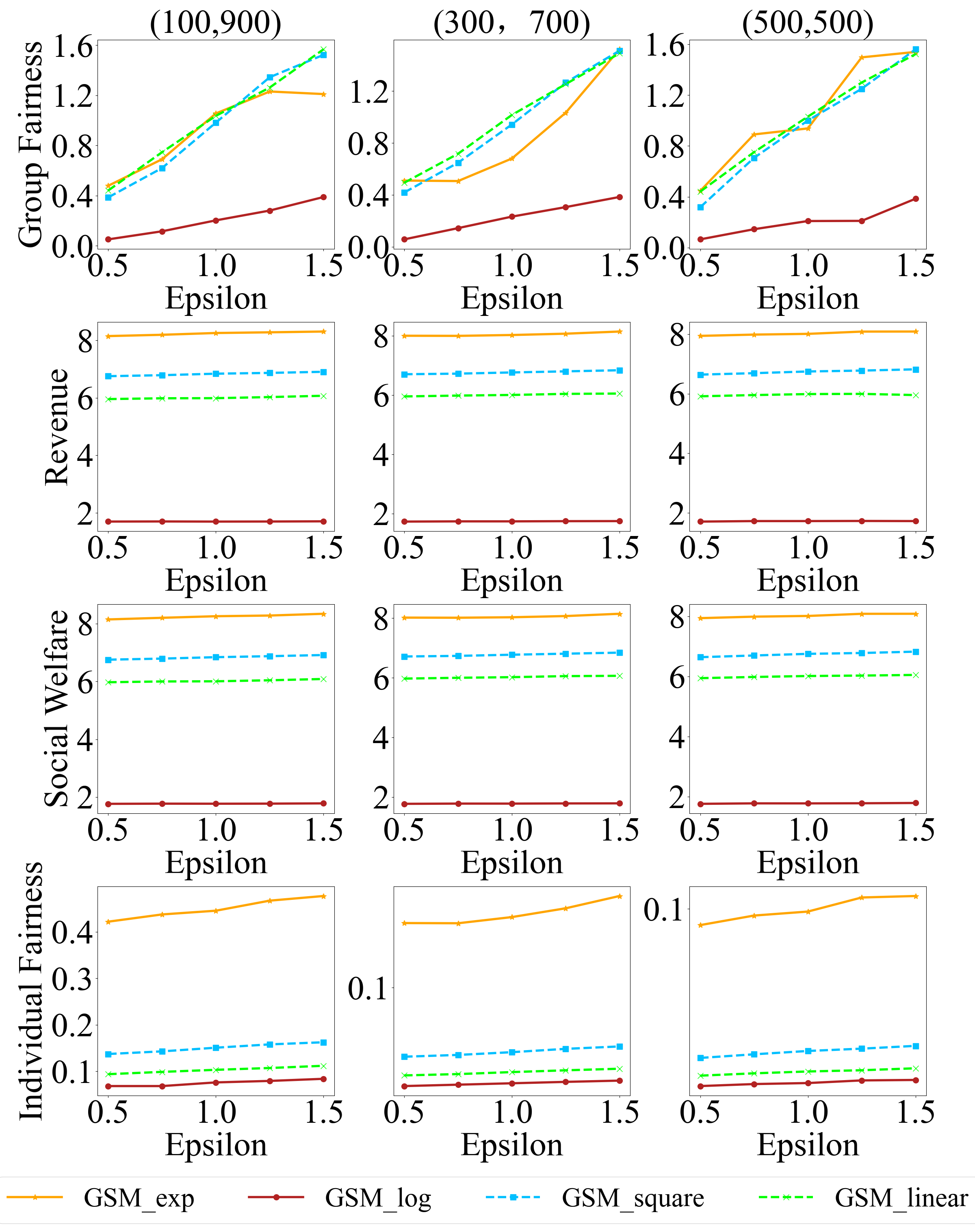}
    \caption{Performance of GSM with linear, square, log, and exponential functions under valuations drawn from $\mathcal{U}(0,10)$ and $\mathcal{U}(0,8)$ and group sizes $(100,900)$ (left), $(300,700)$ (middle), and $(500,500)$ (right).}
    \vspace{0.5cm}
\end{figure}

\subsection{Results} 

Figure~\ref{fig:results-uniform} shows the results under uniformly distributed valuations. For {\bf Q1}, as the value of group fairness increases, meaning the fairness requirement is less strict, both the social welfare and revenue of GPM and GSM increase. When comparing the four mechanisms, the second price auction obtains the highest social welfare and revenue, followed by the simple mechanism and GPM, while GSM obtains the least. It is worth to notice that the cost of obtaining fairness is small; GPM losses only $12.91\%$ of social welfare and $12.94\%$ of revenue comparing to the optimal, while GSM losses $39.21\%$ and $39.38\%$, resp., when the group sizes are $(100,900)$. 

For {\bf Q2}, Fig.~\ref{fig:results-uniform} also shows that GSM obtains almost $0$-individual fairness with $\epsilon^I=0.2$. GPM obtains individual fairness with $\epsilon^I= 5$ 
which is similar to that of simple mechanism.
Such results are consist with our claim that GSM ensures better individual fairness than GPM as the former assigns every buyer a positive winning probability. 
The second-price auction returns the worst 
with $\epsilon^I= 10$. 

For {\bf Q3}, our two mechanisms demonstrate robustness under varying distributed valuations, varying group size differences and varying group fairness levels. {\bf Firstly}, the results under normally distributed valuations show similar patterns as those under uniformly distributed valuations (See Fig.~\ref{fig:results-normal} in the appendix). The losses in social welfare and in revenue are even smaller. This can be attributed to that valuations under a normal distribution tend to be closely clustered, and the payments are closer to the valuations. 
The results for valuations with larger differences 
also show similar patterns (See Fig.~\ref{fig:result-EB} in the appendix). The losses in social welfare and in revenue increase. It might be explained by that larger differences in valuations across group makes it even harder to balance the welfare of the two groups.
{\bf Then} as the differences in group size decreases, the losses of social welfare and of revenue slightly increases in both of our mechanisms. This might be because that when the size of the group with low (high) valuations diminishes (expands), the highest valuation within it tends to decrease (increase), making it easier to balance the welfare of the two groups. 
{\bf Lastly}, with the increase of the value of $\epsilon$, the social welfare and revenue obtained by our two mechanisms increase. It is consistent with our expectation.
The value of individual fairness also increases. It can be explained by that when the group fairness is relaxed, the buyers with high valuations tend to be assigned with higher winning probability, resulting in a decrease in individual fairness. 

For {\bf Q4}, Fig.~\ref{fig:results-function} shows the results for GSMs with different function $f(\theta')$. 
Under the same value of $\epsilon$, GSM with log function has better group fairness than those with other functions. 
Also, in terms of social welfare and revenue, GSM with exponential function yields the highest social welfare and revenue, followed by those with square, linear, and log functions. 
In contrast, GSMs with different functions have an inverse ranking when it comes to individual fairness. 
That is because after the projection, the log function, comparing with the other functions, reduces the differences among the bids and relatively increases the winning chances of buyers with lower bids, which enhances individual fairness but scarifies social welfare and revenue.


\section{Conclusion and future work}
We consider fairness issue in auctions. We propose group fairness as a new concept, and propose two mechanisms that meet 
group fairness and incentive properties. As future work, we plan to explore (1) other auction scenarios, e.g., combinatorial auctions \cite{lehmann2001combinatorial} and auctions over social networks \cite{xiao2022multi,fang2023multi}, and (2) theoretical trade-off between fairness and social welfare, and revenue.



\begin{ack}
This work is partially supported by National Natural Science Foundation of China No. 62172077 and Research Fund for International Senior Scientists No. 62350710215. 
\end{ack}




\bibliography{main}

\begin{thebibliography}{29}
\providecommand{\natexlab}[1]{#1}
\providecommand{\url}[1]{\texttt{#1}}
\expandafter\ifx\csname urlstyle\endcsname\relax
  \providecommand{\doi}[1]{doi: #1}\else
  \providecommand{\doi}{doi: \begingroup \urlstyle{rm}\Url}\fi

\bibitem[off(2016)]{office2016big}
Big data: A report on algorithmic systems, opportunity, and civil rights.
\newblock Technical report, Executive Office of the President, 2016.

\bibitem[Airiau et~al.(2023)Airiau, Gilbert, Grandi, Lang, and Wilczynski]{airiau2023fair}
S.~Airiau, H.~Gilbert, U.~Grandi, J.~Lang, and A.~Wilczynski.
\newblock Fair rent division on a budget revisited.
\newblock In \emph{26th European Conference on Artificial Intelligence (ECAI 2023)}, volume 372, 2023.

\bibitem[Archer and Tardos(2001)]{archer2001truthfulMF}
A.~Archer and {\'E}.~Tardos.
\newblock Truthful mechanisms for one-parameter agents.
\newblock \emph{Proceedings 2001 IEEE International Conference on Cluster Computing}, pages 482--491, 2001.

\bibitem[Aziz and Rey(2021)]{aziz2021almost}
H.~Aziz and S.~Rey.
\newblock Almost group envy-free allocation of indivisible goods and chores.
\newblock In \emph{Proceedings of the Twenty-Ninth International Conference on International Joint Conferences on Artificial Intelligence}, pages 39--45, 2021.

\bibitem[Benabbou et~al.(2019)Benabbou, Chakraborty, Elkind, and Zick]{benabbou2019fairness}
N.~Benabbou, M.~Chakraborty, E.~Elkind, and Y.~Zick.
\newblock Fairness towards groups of agents in the allocation of indivisible items.
\newblock In \emph{Twenty-Eighth International Joint Conference on Artificial Intelligence {IJCAI-19}}, 2019.

\bibitem[Birmpas et~al.(2022)Birmpas, Celli, Colini-Baldeschi, Leonardi, et~al.]{birmpas2022fair}
G.~Birmpas, A.~Celli, R.~Colini-Baldeschi, S.~Leonardi, et~al.
\newblock Fair equilibria in sponsored search auctions: The advertisers’ perspective.
\newblock In \emph{Proceedings of the Thirty-First International Joint Conference on Artificial Intelligence,$\{$IJCAI$\}$ 2022, Vienna, Austria, 23-29 July 2022}, pages 95--101, 2022.

\bibitem[Boyd et~al.(2011)Boyd, Parikh, Chu, Peleato, Eckstein, et~al.]{boyd2011distributed}
S.~Boyd, N.~Parikh, E.~Chu, B.~Peleato, J.~Eckstein, et~al.
\newblock Distributed optimization and statistical learning via the alternating direction method of multipliers.
\newblock \emph{Foundations and Trends{\textregistered} in Machine learning}, 3\penalty0 (1):\penalty0 1--122, 2011.

\bibitem[Budish(2011)]{budish2011combinatorial}
E.~Budish.
\newblock The combinatorial assignment problem: Approximate competitive equilibrium from equal incomes.
\newblock \emph{Journal of Political Economy}, 119\penalty0 (6):\penalty0 1061--1103, 2011.

\bibitem[Caragiannis et~al.(2019)Caragiannis, Kurokawa, Moulin, Procaccia, Shah, and Wang]{caragiannis2019unreasonable}
I.~Caragiannis, D.~Kurokawa, H.~Moulin, A.~D. Procaccia, N.~Shah, and J.~Wang.
\newblock The unreasonable fairness of maximum nash welfare.
\newblock \emph{ACM Transactions on Economics and Computation (TEAC)}, 7\penalty0 (3):\penalty0 1--32, 2019.

\bibitem[Chen and Zhao(2013)]{chen2013bidding}
X.~Chen and J.~Zhao.
\newblock Bidding to drive: Car license auction policy in shanghai and its public acceptance.
\newblock \emph{Transport Policy}, 27:\penalty0 39--52, 2013.

\bibitem[Conitzer et~al.(2019)Conitzer, Freeman, Shah, and Vaughan]{conitzer2019group}
V.~Conitzer, R.~Freeman, N.~Shah, and J.~W. Vaughan.
\newblock Group fairness for the allocation of indivisible goods.
\newblock \emph{Proceedings of the AAAI Conference on Artificial Intelligence}, 33:\penalty0 1853--1860, 2019.

\bibitem[Dressel and Farid(2018)]{dressel2018accuracy}
J.~Dressel and H.~Farid.
\newblock The accuracy, fairness, and limits of predicting recidivism.
\newblock \emph{Science advances}, 4\penalty0 (1):\penalty0 eaao5580, 2018.

\bibitem[Dwork et~al.(2012)Dwork, Hardt, Pitassi, Reingold, and Zemel]{dwork2012fairness}
C.~Dwork, M.~Hardt, T.~Pitassi, O.~Reingold, and R.~Zemel.
\newblock Fairness through awareness.
\newblock In \emph{Proceedings of the 3rd innovations in theoretical computer science conference}, pages 214--226, 2012.

\bibitem[Eggum and Larsen(2021)]{eggum2021housing}
T.~Eggum and E.~R. Larsen.
\newblock Is the housing market an inequality generator?
\newblock \emph{URL https://housinglab. oslomet. no/wp-content/uploads/2021/06/HLWP2021\_ 2\_new. pdf}, 2021.

\bibitem[Fang et~al.(2023)Fang, Zhang, Liu, Khoussainov, and Xiao]{fang2023multi}
Y.~Fang, M.~Zhang, J.~Liu, B.~Khoussainov, and M.~Xiao.
\newblock Multi-unit auction over a social network.
\newblock In \emph{ECAI 2023}, pages 676--683. IOS Press, 2023.

\bibitem[Finocchiaro et~al.(2021)Finocchiaro, Maio, Monachou, Patro, Raghavan, Stoica, and Tsirtsis]{finocchiaro2021bridging}
J.~Finocchiaro, R.~Maio, F.~Monachou, G.~K. Patro, M.~Raghavan, A.-A. Stoica, and S.~Tsirtsis.
\newblock Bridging machine learning and mechanism design towards algorithmic fairness.
\newblock In \emph{Proceedings of the 2021 ACM Conference on Fairness, Accountability, and Transparency}, pages 489--503, 2021.

\bibitem[Gal et~al.(2016)Gal, Mash, Procaccia, and Zick]{gal2016fairest}
Y.~Gal, M.~Mash, A.~D. Procaccia, and Y.~Zick.
\newblock Which is the fairest (rent division) of them all?
\newblock In \emph{Proceedings of the 2016 ACM Conference on Economics and Computation}, pages 67--84, 2016.

\bibitem[Gamov and Stern(1958)]{gamov1958puzzle}
G.~Gamov and M.~Stern.
\newblock Puzzle-math.
\newblock \emph{Viking, New York}, 1958.

\bibitem[Grgic-Hlaca et~al.(2016)Grgic-Hlaca, Zafar, Gummadi, and Weller]{grgic2016case}
N.~Grgic-Hlaca, M.~B. Zafar, K.~P. Gummadi, and A.~Weller.
\newblock The case for process fairness in learning: Feature selection for fair decision making.
\newblock In \emph{NIPS symposium on machine learning and the law}, volume~1, page~11. Barcelona, Spain, 2016.

\bibitem[Hardt et~al.(2016)Hardt, Price, and Srebro]{hardt2016equality}
M.~Hardt, E.~Price, and N.~Srebro.
\newblock Equality of opportunity in supervised learning.
\newblock \emph{Advances in neural information processing systems}, 29, 2016.

\bibitem[Kusner et~al.(2017)Kusner, Loftus, Russell, and Silva]{kusner2017counterfactual}
M.~J. Kusner, J.~Loftus, C.~Russell, and R.~Silva.
\newblock Counterfactual fairness.
\newblock \emph{Advances in neural information processing systems}, 30, 2017.

\bibitem[Lambrecht and Tucker(2019)]{lambrecht2019algorithmic}
A.~Lambrecht and C.~Tucker.
\newblock Algorithmic bias? an empirical study of apparent gender-based discrimination in the display of stem career ads.
\newblock \emph{Management science}, 65\penalty0 (7):\penalty0 2966--2981, 2019.

\bibitem[Lehmann et~al.(2001)Lehmann, Lehmann, and Nisan]{lehmann2001combinatorial}
B.~Lehmann, D.~Lehmann, and N.~Nisan.
\newblock Combinatorial auctions with decreasing marginal utilities.
\newblock In \emph{Proceedings of the 3rd ACM conference on Electronic Commerce}, pages 18--28, 2001.

\bibitem[Lipton et~al.(2004)Lipton, Markakis, Mossel, and Saberi]{lipton2004approximately}
R.~J. Lipton, E.~Markakis, E.~Mossel, and A.~Saberi.
\newblock On approximately fair allocations of indivisible goods.
\newblock In \emph{Proceedings of the 5th ACM Conference on Electronic Commerce}, pages 125--131, 2004.

\bibitem[Rashid et~al.(2020)Rashid, Samvelyan, De~Witt, Farquhar, Foerster, and Whiteson]{rashid2020monotonic}
T.~Rashid, M.~Samvelyan, C.~S. De~Witt, G.~Farquhar, J.~Foerster, and S.~Whiteson.
\newblock Monotonic value function factorisation for deep multi-agent reinforcement learning.
\newblock \emph{The Journal of Machine Learning Research}, 21\penalty0 (1):\penalty0 7234--7284, 2020.

\bibitem[Steinhaus(1949)]{steinhaus1949division}
H.~Steinhaus.
\newblock Sur la division pragmatique.
\newblock \emph{Econometrica: Journal of the Econometric Society}, pages 315--319, 1949.

\bibitem[Tan and Wei(2020)]{tan2020evaluating}
L.~Tan and L.~Wei.
\newblock Evaluating car license auction mechanisms: theory and experimental evidence.
\newblock \emph{China Economic Review}, 60:\penalty0 101387, 2020.

\bibitem[Todo et~al.(2011)Todo, Li, Hu, Mouri, and Yokoo]{todo2011generalizing}
T.~Todo, R.~Li, X.~Hu, T.~Mouri, and M.~Yokoo.
\newblock Generalizing envy-freeness toward group of agents.
\newblock In \emph{IJCAI 2011, Proceedings of the 22nd International Joint Conference on Artificial Intelligence, Barcelona, Catalonia, Spain, July 16-22, 2011}, 2011.

\bibitem[Xiao et~al.(2022)Xiao, Song, and Khoussainov]{xiao2022multi}
M.~Xiao, Y.~Song, and B.~Khoussainov.
\newblock Multi-unit auction in social networks with budgets.
\newblock In \emph{Proceedings of the AAAI Conference on Artificial Intelligence}, volume~36, pages 5228--5235, 2022.

\end{thebibliography}

\clearpage

\appendix

\noindent{\LARGE \bf Appendix}

\section{Experiments} \label{app:exp}
Here we present the complementary results. Figure~\ref{fig:results-normal} shows the performance of the four mechanisms under the valuations drawn from normal distributions $\mathcal{N}(5,1)$ and $\mathcal{N}(4,1)$. Figures~\ref{fig:result-EB} shows the performance under the valuations drawn from uniform $\mathcal{U}(0,10) \& \mathcal{U}(0,4)$ and $\mathcal{N}(5,1) \& \mathcal{N}(2,1)$. 

\begin{figure}[ht]
    \centering
    \includegraphics[width=1\columnwidth]{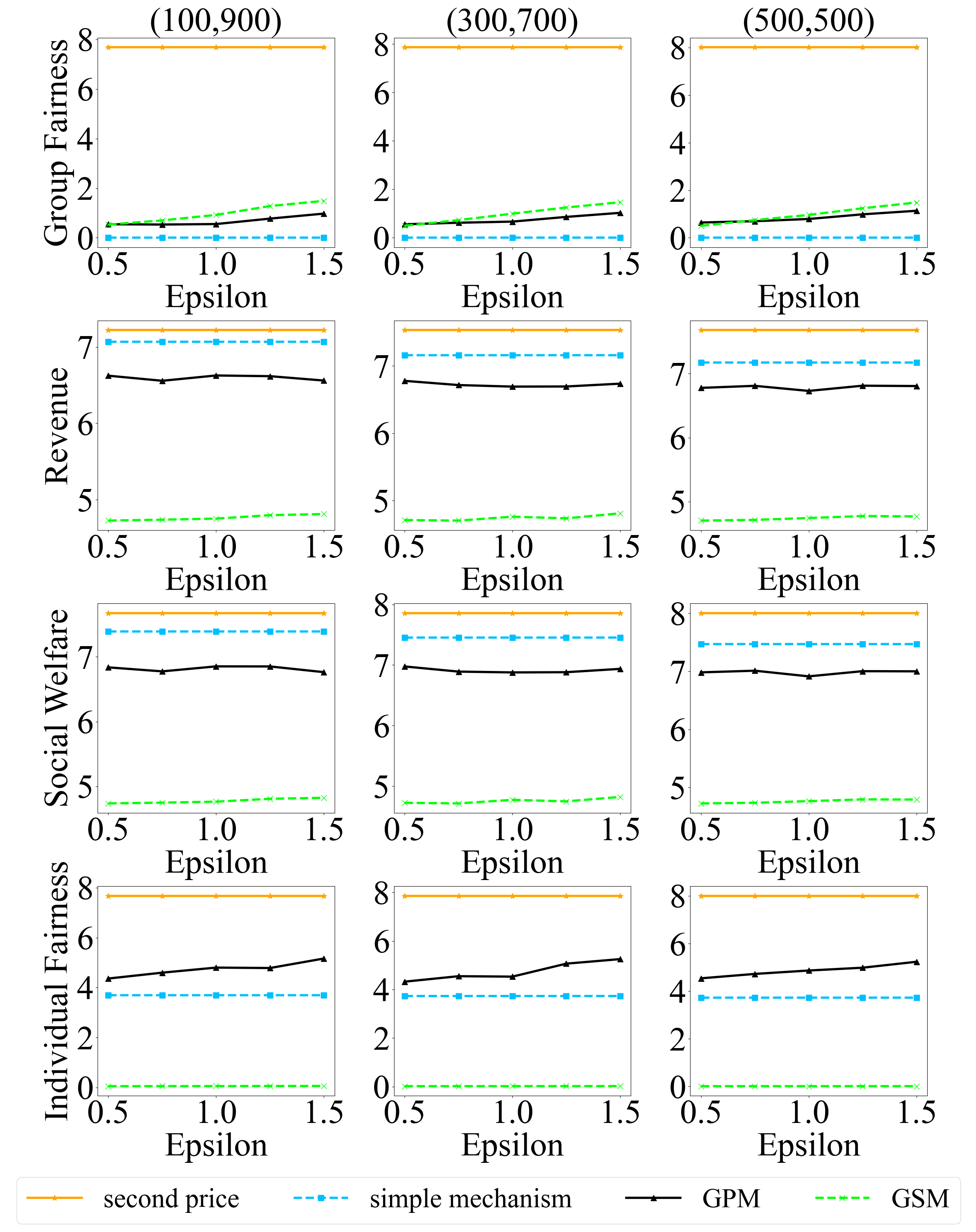}
    \caption{Performance of GPM, GSM, second-price auction and simply mechanism under normally distributed valuations and group sizes $(100,900)$ (left), $(300,700)$ (middle), and $(500,500)$ (right). }
    \label{fig:results-normal}
\end{figure}

\begin{figure}[ht]
    \centering
    \includegraphics[width=\columnwidth]{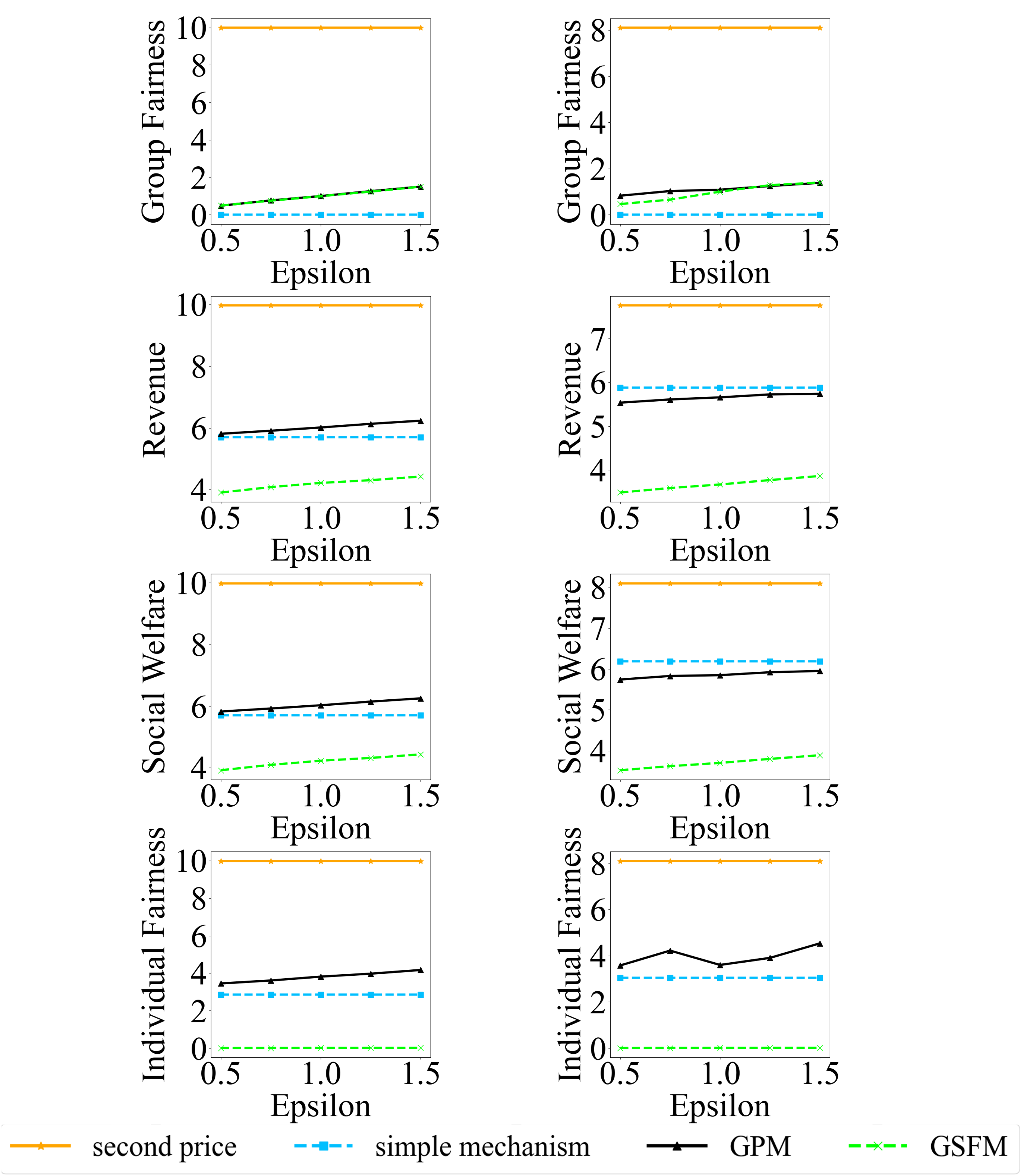}
    \caption{Performance of GPM, GSM with linear function, second-price auction and simply mechanism under group sizes $(500,500)$ and under valuations drawn from $\mathcal{U}(0,10) \& \mathcal{U}(0,4)$ (left), $\mathcal{N}(5,1) \& \mathcal{N}(2,1)$ (right). }
    \label{fig:result-EB} 
\end{figure}

%

\end{document}